\documentclass[12pt]{revtex4}
\usepackage{epstopdf}
\expandafter\let\csname equation*\endcsname\relax
\expandafter\let\csname endequation*\endcsname\relax
\usepackage{amsmath}
\usepackage{amsfonts}
\usepackage{amssymb}
\usepackage{makeidx}

\newcommand{\beq}{\begin{equation}} \newcommand{\eeq}{\end{equation}}
\newcommand{\bea}{\begin{eqnarray}} \newcommand{\eea}{\end{eqnarray}}
\newcommand{\bear}{\begin{eqnarray*}} \newcommand{\eear}{\end{eqnarray*}}
\newcommand{\lb}{\label}
\newcommand{\rf}[1]{(\ref{#1})}

\newtheorem{theorem}{Theorem}[section]

\newtheorem{definition}{Definition}
\newtheorem{remark}{Remark}

\newenvironment{proof}[1][Proof]{\begin{trivlist}
\item[\hskip \labelsep {\bfseries #1}]}{\end{trivlist}}

\begin{document}

\title {An Action Principle for Action-dependent Lagrangians: toward an Action Principle to non-conservative systems}

\author{Matheus J. Lazo$^{1}$\footnote{matheuslazo@furg.br}, Juilson Paiva$^1$, Jo\~ao T. S. Amaral$^1$, and Gast\~ao S. F. Frederico$^2$}

\affiliation{$^1$Universidade Federal do Rio Grande, Rio Grande, RS, Brazil.\\$^2$Universidade Federal do Cear\'a, Campus de Russas, Russas, Brazil.}

\begin{abstract}

In this work, we propose an Action Principle for Action-dependent Lagrangian functions by generalizing the Herglotz variational problem to the case with several independent variables. We obtain  a necessary condition for the extremum equivalent to the Euler-Lagrange equation and, through some examples, we show that this generalized Action Principle enables us to construct simple and physically meaningful Action-dependent Lagrangian functions for a wide range of non-conservative classical and quantum systems. Furthermore, when the dependence on the Action is removed, the traditional Action Principle for conservative systems is recovered.

%{\it PACS}: 

%\keywords{Action Principle, Lagrangians for dissipative systems}

\end{abstract}

\maketitle

%%%%%%%%%%%%%%%%%%%%%%%%%%%%%%%%%%%%%%%%%%%%%%%%%%%%%%%%%%%%%%%%%%%%%%%%%%%%%%%%%%%%%%%%%%%%

\section{Introduction}

The Action Principle was introduced in its mature formulation by Euler, Lagrange, and Hamilton and, since then, it has become one of the most fundamental principles of physics. In any (classical or quantum) physical theory, the dynamics of a conservative system is always given by the Action Principle. However, it is well known that the equation of motion for dissipative linear dynamical systems with constant coefficients cannot be obtained by the variational principle. A rigorous proof for the failure of the Action Principle in describing non-conservative systems was given in 1931 by Bauer \cite{bauer}, who proved that, from the traditional Action Principle, it is impossible to obtain a dissipation term proportional to the first order time derivative in the equation of motion. In order to deal with this difficulty, over the last century, several methods have been developed. Examples include time-dependent Lagrangians \cite{Stevens}, the Bateman approach by introducing auxiliary coordinates that describe the reverse-time system \cite{Morse} and Actions with fractional derivatives \cite{Riewe,LazoCesar}. Unfortunately, these approaches either give us non-physical Lagrangian functions (in the sense that they provide non-physical relations for the momentum and Hamiltonian of the system) or make use of non-local differential operators with algebraic properties different from usual derivatives (see \cite{Riewe,LazoCesar} for a detailed discussion).

In order to formulate an Action Principle for non-conservative systems, we take a different approach by generalizing the Herglotz variational problem \cite{Herg1,Herg2,GGB} and the ideas introduced in our recent work \cite{MJJG}.
In any physical theory, the Lagrangian function which defines the Action is constructed from the scalars of the theory, and from it, the corresponding dynamical equations can be obtained. However, the Action itself is a scalar and we might ask ourselves what would happen if the Lagrangian function itself were a function of the Action. For a one dimensional system, the answer to this question can be given by an almost forgotten variational problem proposed by Herglotz in 1930 \cite{Herg1,Herg2,GGB}. A reason for this problem to be almost unknown is that a covariant generalization for several independent variables is not direct and still lacks. In a recent work \cite{MJJG}, we generalized the Herglotz problem to construct a non-conservative gravitational theory from the Lagrangian formalism. However, the results presented in \cite{MJJG} are restricted only to gravitation, and a necessary condition (equivalent to the Euler-Lagrange equation) for this covariant variational problem has not yet been obtained.

In the present work, we formulate a generalization of the Action Principle introduced in \cite{MJJG} for arbitrary Action-dependent Lagrangian functions, and we obtain a generalized Euler-Lagrange equation for the problem. Both the Action Principle and the corresponding necessary condition reduce to the classical ones when the Lagrangian function does not depend on the Action. Furthermore, in order to investigate the potential of application of this generalized Action Principle to study non-conservative systems, we show, through four examples, that the simplest case of a Lagrangian linear on the Action gives us the correct equation of motion for both classical and quantum dissipative systems. In special, we consider a classical vibrating string under viscous forces, a non-conservative Electromagnetic Theory, and dissipative Schr\"{o}dinger and Klein-Gordon equations.

The paper is organized as follows. In Section \ref{sec2} we show in detail how the Herglotz variational principle can be generalized to the case involving several independent variables. A particularly important consequence of this is the resulting set of the generalized Euler-Lagrange equations of motion. Section \ref{sec3} is devoted to applications of the generalized Action principle to some (classical and quantum) non-conservative systems, described by Lagrangian functions which depend on one or more fields. We conclude in Section \ref{sec4}.

%%%%%%%%%%%%%%%%%%%%%%%%%%%%%%%%%%%%%%%%%%%%%%%%%%%%%%%%%%%%%%%%%%%%%%%%%%%%%%%%%%%%%%%%%%

\section{Action Principle for Action-dependent Lagrangians}\label{sec2}

The variational problem proposed by Herglotz in 1930 \cite{Herg1,Herg2} consists in the problem of determining the path $x(t)$ that  extremizes (minimizes or maximizes) $S(b)$, where $S(t)$ is a solution of
\beq
\lb{H}
\begin{split}
\dot{S}(t)&=L(t,x(t),\dot{x}(t),S(t)),\;\;\; t\in [a,b]\\
S(a)&=s_a, \;\;\; x(a)=x_a, \;\;\; x(b)=x_b, \;\;\; s_a,x_a,x_b \in \mathbb{R}.
\end{split}
\eeq
It is important to notice that \rf{H} represents a family of differential equations since for each function $x(t)$ a different differential equation arises. Therefore, $S(t)$ depends on $x(t)$. The problem reduces to the classical fundamental problem of the calculus of variations if the Lagrangian function $L$ does not depend on $S(t)$. In this case, we have:
\beq
\lb{H1}
\begin{split}
\dot{S}(t)&=L(t,x(t),\dot{x}(t)),\;\;\; t\in [a,b]\\
S(a)&=s_a, \;\;\;  x(a)=x_a, \;\;\; x(b)=x_b, \;\;\; s_a,x_a,x_b \in \mathbb{R},
\end{split}
\eeq
then, by integrating \eqref{H1}, we obtain the classical variational problem
\beq
\lb{H2}
S(b)=\int_a^b \tilde{L}(t,x(t),\dot{x}(t))\;dt\longrightarrow {\mbox{extremum}},
\eeq
where $x(a)=x_a$, $x(b)=x_b$, and
\beq
\lb{H3}
\tilde{L}(t,x(t),\dot{x}(t))=L(t,x(t),\dot{x}(t)) +\frac{s_a}{b-a}.
\eeq
For what follows it is important to notice, from \eqref{H2}, that for a given function $x(t)$ the functional $S$
reduces to a function on the boundaries $a,b$ of the domain $[a,b]$.

Herglotz proved \cite{Herg1,Herg2} that a necessary condition for a path $x(t)$ to imply an extremum of the variational problem \eqref{H} is given by the generalized Euler-Lagrange equation:
\beq
\lb{HEL}
\frac{\partial L}{\partial x} -\frac{d}{dt}\frac{\partial L}{\partial \dot{x}}+\frac{\partial L}{\partial S}\frac{\partial L}{\partial \dot{x}}=0.
\eeq
It should be noticed that in the case of the classical problem of the calculus of variation \eqref{H2}, we have $\frac{\partial L}{\partial S}=0$, and the differential equation \eqref{HEL} reduces to the classical Euler-Lagrange equation. Moreover, the application of Herglotz problem to non-conservative systems is evident even in the simplest case, where the dependence of the Lagrangian function on the Action is linear. For example, the function
\beq
\lb{H3b}
L=\frac{m\dot{x}^ 2}{2}-U(x)-\frac{\gamma}{m} S
\eeq
describes a dissipative system with a point particle of mass $m$ under a potential $U(x)$ and a viscous force with a resistance coefficient $\gamma$. From \eqref{HEL}, the resulting equation of motion
\beq
\lb{H3c}
m\ddot{x}+\gamma \dot{x}=F
\eeq
includes the well-known dissipative term proportional to the velocity $\dot{x}$, where $\ddot{x}$ is the particle acceleration and $F=-\frac{d U}{dx}$ is the external force. In this context, the linear term $\frac{\gamma}{m} S$ in the Lagrangian function \eqref{H3b} can be interpreted as a potential function for the non-conservative force. Furthermore, the Lagrangian given by \eqref{H3b} is physical in the sense it provides us with physically meaningful relations for the momentum and the Hamiltonian \cite{Riewe,LazoCesar}. If we define the canonical variables
\beq
\lb{H3d}
q=x, \qquad p=\frac{\partial L}{\partial \dot{q}}=m\dot{x},
\eeq
we obtain the Hamiltonian
\beq
\lb{H3e}
H=qp-L=\frac{m\dot{x}^ 2}{2}+U(x)+\frac{\gamma}{m} S.
\eeq
From \eqref{H3d} and \eqref{H3e}, we can see that the Lagrangian \eqref{H3b} is physical in the sense it provides us a correct relation for the momentum $p= m\dot{x}$, and a physically meaningful Hamiltonian given by the sum of all energies.

On the other hand, for a scalar field $\phi(x^{\mu})=\phi(x^1,x^2,\cdots,x^d)$ defined in a domain $\Omega \in \mathbb{R}^d$ ($d=1,2,3,\cdots$), the classical problem of variation calculus deals with the problem to find $\phi$ that extremizes the functional
\beq
\lb{H4}
S(\delta\Omega)=\int_{\delta\Omega}\mathcal{L}\left(x^\mu,\phi(x^\mu),\partial_\nu\phi(x^\mu)\right)d^dx,
\eeq
where $\delta\Omega$ is the boundary of $\Omega$, and $\phi$ satisfies the boundary condition $\phi(\delta\Omega)=\phi_{\delta\Omega}$ with $\phi_{\Omega}:\delta \Omega \longrightarrow \mathbb{R}^n$. Unfortunately, although the Herglotz problem was introduced in 1930, a covariant generalization of \eqref{H} for several independent variables is not direct and still lacks. The cornerstone of a generalization of the Herglotz problem for fields is to note that, as in \eqref{H2}, for a given fixed $\phi$ the functional $S$, defined in \eqref{H4}, reduces to a function of the boundary $\delta \Omega$. In this context, if there is a differentiable vector field $s^{\mu}$ such that
\beq
\lb{H5}
S(\delta\Omega)=\int_{\delta\Omega}s^\mu n_\mu\; d\sigma,
\eeq
then, from the Divergence Theorem we obtain
\beq
\lb{H6}
S(\delta\Omega)=\int_{\delta\Omega}s^\mu n_\mu d\sigma=\int_{\Omega}\partial_\mu s^\mu d^dx=\int_{\Omega}\mathcal{L}\left(x^\mu,\phi(x^\mu),\partial_\nu\phi(x^\mu),s^\mu\right) d^dx,
\eeq
where we consider that $\delta\Omega$ is an orientable Jordan surface, $n_{\mu}$ is a normal vector to it and $d\sigma$ is the surface differential. Consequently, we can generalize the Herglotz variational principle as follows:

\begin{definition}[Fundamental Problem]
\lb{HF}
Let the action-density field $s^{\mu}$ be a differentiable vector field on $\Omega \in \mathbb{R}^d$. The fundamental problem of Herglotz variational principle for fields consists in determining the field $\phi$ that  extremizes (minimizes or maximizes) $S(\delta\Omega)$, where $S(\delta\Omega)$ is given by
\beq
\lb{H7}
\begin{split}
&\partial_\mu s^\mu=\mathcal{L}\left(x^\mu,\phi(x^\mu),\partial_\mu\phi(x^\mu),s^\mu\right),\quad x^\mu=(x^1, x^2,..., x^d)\in\Omega\\
&S(\delta\Omega)=\int_{\delta\Omega}s^\mu n_\mu d\sigma,\quad\phi(\delta\Omega)=\phi_{\delta\Omega},\quad \phi(\delta\Omega):\delta\Omega\longrightarrow \mathbb{R}.
\end{split}
\eeq
\end{definition}

Like in the original Herglotz problem, it is easy to notice that our Action functional defined by \eqref{H7} reduces to the usual Action \eqref{H4} when the Lagrangian function is independent of the action-density field $s^{\mu}$. Furthermore, we can prove the following condition for the extremum of \eqref{H7}:

\begin{theorem}[Generalized Euler-Lagrange equation for non-conservative fields]
\lb{HELF}
Let $\partial_{s^\mu}\mathcal{L}=\gamma_\mu$ be a gradient $\gamma_\mu=\partial_{\mu}f(x^\nu)=(\partial_{x_1}f,\cdots,\partial_{x_d}f)$ of a scalar field $f:\Omega \longrightarrow \mathbb{R}$, and let $\phi^*$ be the fields that extremize $S(\delta\Omega)$ defined in \eqref{H7}. Then, the field $\phi^*$ satisfies the generalized Euler-Lagrange equation
\beq
\label{GHEL}
\dfrac{\partial\mathcal{L}}{\partial \phi^* }-\frac{d}{dx^\mu}\left(\dfrac{\partial \mathcal{L} }{\partial\left(\partial_\mu\phi^*\right)}\right)+\gamma_\mu \dfrac{\partial \mathcal{L} }{\partial\left(\partial_\mu\phi^*\right)}=0.
\eeq
\end{theorem}

\begin{proof}
Let us define a family of fields $\phi$ (weak variations) such that
\beq
\lb{PH1}
\phi(x^\mu)=\phi^*(x^\mu)+\varepsilon\eta(x^\mu),
\eeq
where $\varepsilon \in \mathbb{R}$ and $\eta(x^\mu)$ is a field satisfying the boundary condition $\eta(\delta\Omega)=0$.
Now, after integrating both sides of the differential equation in \eqref{H7} over $\Omega$ we get
\beq
\lb{PH5}
S(\delta\Omega)=\int_\Omega \mathcal{L}\left(x^\mu,\phi(x^\mu),\partial_\nu\phi(x^\mu),s^\mu\right)d^dx,
%S(\delta\Omega)=\int_{\Omega}L(\mathbf{x},\phi,\nabla\phi,\mathbf{s})\;d^dx.
\eeq
and, taking the derivative with respect to $\varepsilon$ we obtain the following relation
\beq
\lb{PH6}
\begin{split}
\dfrac{dS(\delta\Omega)}{d\varepsilon}&=\int_\Omega\dfrac{d}{d\varepsilon}\mathcal{L}\left(x^\mu,\phi(x^\mu),\partial_\mu\phi(x^\mu),s^\mu\right)d^dx
\\&=\int_\Omega\left[\eta\dfrac{\partial \mathcal{L}}{\partial\phi}+\partial_\mu\eta\dfrac{\partial \mathcal{L}}{\partial\left(\partial_\mu\phi\right)}+\gamma_\mu\dfrac{ds^\mu}{d\varepsilon}\right]d^dx,
\end{split}
\eeq
where for simplicity we write $\frac{\partial \mathcal{L}}{\partial (\partial \phi)}=\left(\frac{\partial\mathcal{L}}{\partial (\partial_1 \phi)},\frac{\partial \mathcal{L}}{\partial (\partial_2 \phi)},...,\frac{\partial \mathcal{L}}{\partial (\partial_d \phi)}\right)$.
On the other hand, we also have from \eqref{H7}
\beq
\lb{PH7}
\dfrac{dS(\delta\Omega)}{d\varepsilon}=\int_{\delta\Omega}\dfrac{ds^\mu}{d\varepsilon}n_\mu d\sigma=\int_\Omega\partial_\mu\dfrac{ds^\mu}{d\varepsilon}d^dx.
\eeq
By inserting \eqref{PH7} into \eqref{PH6} we get
\beq
\lb{PH8}
\int_\Omega\left[\partial_\mu\dfrac{ds^\mu}{d\varepsilon}-\eta\dfrac{\partial \mathcal{L}}{\partial\phi}-\partial_\mu\eta\dfrac{\partial \mathcal{L}}{\partial(\partial_\mu\phi)}-\gamma_\mu\dfrac{ds^\mu}{d\varepsilon}\right]d^dx=0.
\eeq
A sufficient condition to satisfy \eqref{PH8} for any domain $\Omega$ is
\beq
\lb{PH9}
\partial_\mu\zeta^\mu-\eta\dfrac{\partial \mathcal{L}}{\partial\phi}-\partial_\mu\eta\dfrac{\partial \mathcal{L}}{\partial(\partial_\mu\phi)}-\gamma_\mu\zeta^\mu=0,
\eeq
where $\zeta^\mu=\dfrac{ds^\mu}{d\varepsilon}$. Since $\gamma_\mu=\partial_\mu f(x^\nu)$ is a gradient vector on $\Omega$, \eqref{PH9} implies that $\zeta$ can be written as
\beq
\lb{PH10}
\zeta^\mu(\varepsilon)=A^\mu\left(x^\mu,\phi,\partial_\mu\phi,s^\mu\right)e^{f(x^\nu)},
\eeq
where
\beq
\lb{PH11}
\partial_\mu A^\mu\left(x^\mu,\phi,\partial_\mu\phi,s^\mu\right)=\left(\eta\dfrac{\partial \mathcal{L}}{\partial\phi}+\partial_\mu\eta\dfrac{\partial  \mathcal{L}}{\partial(\partial_\mu\phi)}\right)e^{-f(x^\nu)}.
\eeq

Now, since $S(\delta\Omega)$ attains a maximum (minimum) at $\phi^*$, we should have
\beq
\lb{PH2}
\frac{d S(\delta\Omega)}{d\varepsilon}\vert_{\varepsilon=0}=0.
\eeq
Then, from \rf{H7} and \eqref{PH2} we get
\beq
\lb{PH3}
\dfrac{dS(\delta\Omega)}{d\varepsilon}\bigg|_{\varepsilon=0}=\int_{\delta\Omega}\dfrac{ds^\mu}{d\varepsilon}n_\mu \bigg|_{\varepsilon=0}d\sigma=\int_{\delta\Omega}\zeta^\mu(0)n_\mu d\sigma=0,
\eeq
since the surface $\delta\Omega$ is independent on $\varepsilon$. Furthermore, since \eqref{PH3} should hold for any domain, we have
\beq
\lb{PH12}
\int_{\delta\Omega}\zeta^\mu(0)n_\mu d\sigma=0 \Longrightarrow\int_{\delta\Omega}A^\mu n_\mu \bigg|_{\varepsilon=0}d\sigma=\int_\Omega\partial_\mu A^\mu\bigg|_{\varepsilon=0}d^dx=0.
\eeq
Thus from \eqref{PH11} and \eqref{PH12} we get
\beq
\lb{PH13}
\begin{split}
&\int_\Omega\left[\eta\frac{\partial \mathcal{L}\left(x^\mu,\phi^*,\partial_\mu\phi^*,s^\mu\right)}{\partial\phi^*}+\partial_\mu\eta\frac{\partial \mathcal{L}\left(x^\mu,\phi^*,\partial_\mu\phi^*,s^\mu\right)}{\partial(\partial_\mu\phi^*)}\right]e^{-f(x^\nu)} d^dx\\
&=\int_\Omega\left[\frac{\partial \mathcal{L}\left(x^\mu,\phi^*,\partial_\mu\phi^*,s^\mu\right)}{\partial\phi^*}-\frac{d}{dx^\mu}\frac{\partial \mathcal{L}\left(x^\mu,\phi^*,\partial_\mu\phi^*,s^\mu\right)}{\partial(\partial_\mu\phi^*)}\right.\\
&\qquad\qquad\qquad\qquad\qquad\qquad\qquad+\left.\gamma_\mu\frac{\partial \mathcal{L}\left(x^\mu,\phi^*,\partial_\mu\phi^*,s^\mu\right)}{\partial\left(\partial_\mu\phi^*\right)}\right]e^{-f(x^\nu)}\eta d^dx
\\&+\int_\Omega\left[\frac{d}{dx^\mu}\left(\eta\frac{\partial  \mathcal{L}\left(x^\mu,\phi^*,\partial_\mu\phi^*,s^\mu\right)}{\partial(\partial_\mu\phi^*)}e^{-f(x^\nu)}\right)\right]d^dx=0.
\end{split}
\eeq
The last integral in \eqref{PH13} is zero since $\eta(\delta \Omega)=0$. Thus, from the Fundamental Lemma of calculus of variation we obtain \eqref{GHEL}.
\end{proof}

Before enunciating our Action Principle for fields and considering some applications, some remarks concerning the sufficient condition \eqref{PH9} and some particular cases, are in order:

\begin{remark}
In the proof of Theorem \ref{HELF}, in order to satisfy condition \eqref{PH8} for any domain $\Omega$, we impose the sufficient condition \eqref{PH9}. However, this is not the only possibility to satisfy \eqref{PH8}. Another way is to integrate the third term in \eqref{PH8} by parts, and the result is
\beq
\lb{PH9B}
\partial_\mu\zeta^\mu-\eta\dfrac{\partial \mathcal{L}}{\partial\phi}+\eta\frac{d}{dx^{\mu}}\dfrac{\partial \mathcal{L}}{\partial(\partial_\mu\phi)}-\gamma_\mu\zeta^\mu=0,
\eeq
instead of \eqref{PH9}. We discard this possibility since, by following exactly the same development in the proof, the traditional Euler-Lagrange equation for conservative systems is obtained,
\beq
\label{GHELB}
\dfrac{\partial\mathcal{L}}{\partial \phi^* }-\frac{d}{dx^\mu}\left(\dfrac{\partial \mathcal{L} }{\partial\left(\partial_\mu\phi^*\right)}\right)=0,
\eeq
instead of \eqref{GHEL}. Since \eqref{GHELB} implies that \eqref{PH9B} reduces to $\partial_\mu\zeta^\mu-\gamma_\mu\zeta^\mu=0$, then we get $\zeta^\mu=A^\mu e^{f(x^\nu)}$, where now $A^\mu$ is a constant. However, the condition \eqref{PH3} implies that $A^\mu=0$, resulting in $\zeta^\mu=0$ and $s^\mu$ independent on $\varepsilon$. As a consequence, we have two possibilities in this case: either the functional $S(\delta \Omega)$ and $s^\mu$ are independent of the field $\phi$ or the Lagrangian is independent of $s^\mu$. The first case in not interesting since the functional is a constant, and the second case is the classical problem of variation calculus.
\end{remark}

\begin{remark}
It is easy to see that for Lagrangian functions independent on $s^{\mu}$, the generalized Euler-Lagrange equation \rf{GHEL} reduces to the usual one,
\beq
\lb{HEL2}
\dfrac{\partial\mathcal{L}}{\partial \phi^* }-\frac{d}{dx^\mu}\left(\dfrac{\partial \mathcal{L} }{\partial\left(\partial_\mu\phi^*\right)}\right)=0,
\eeq
since, in this case, $\gamma_\mu =0$.
\end{remark}

\begin{remark} When the action-density field $s^{\mu}$ has only one non-null component and it is a function of only one variable, for example $s^1\neq 0$ and $x^1=t$, and $\Omega=[t_a,t_b]\otimes \mathbb{R}^{d-1}$, the fundamental problem in Definition \ref{HF} contains, as a particular case, the non-covariant problem introduced in \cite{GGB}. Moreover, in the latter situation , equation \eqref{H7} can be easily solved for Lagrangian functions linear on $s^1$, resulting in a $s^1$ expressed as a history-dependent function.

\end{remark}

\begin{remark}
It is straightforward to generalize  the fundamental problem \eqref{H7} and the Euler-Lagrange equation \eqref{GHEL} to the case with several fields $\phi^i(x^{\mu})=\phi^i(x^1,x^2,\cdots,x^d)$ ($i=1,...,N$). In this case we have the Action $S(\delta\Omega)$ defined by
\beq
\lb{H7B}
\begin{split}
&\partial_\mu s^\mu=\mathcal{L}\left(x^\mu,\phi^i(x^\mu),\partial_\mu\phi^i(x^\mu),s^\mu\right),\quad x^\mu=(x^1, x^2,..., x^d)\in\Omega\\
&S(\delta\Omega)=\int_{\delta\Omega}s^\mu n_\mu d\sigma,\quad\phi^i(\delta\Omega)=\phi^i_{\delta\Omega},\quad \phi^i(\delta\Omega):\delta\Omega\longrightarrow \mathbb{R},
\end{split}
\eeq
and for a Lagrangian function for which $\partial_{s^\mu}\mathcal{L}=\gamma_\mu=\partial_{\mu}f(x^\nu)$, we obtain the following set of generalized Euler-Lagrange equations:
\beq
\label{GHEL2}
\dfrac{\partial\mathcal{L}}{\partial \phi^{i*} }-\frac{d}{dx^\mu}\left(\dfrac{\partial \mathcal{L} }{\partial\left(\partial_\mu\phi^{i*}\right)}\right)+\gamma_\mu \dfrac{\partial \mathcal{L} }{\partial\left(\partial_\mu\phi^{i*}\right)}=0,\quad i=1,...,N.
\eeq
\end{remark}

We can now formulate an Action Principle suited to dissipative systems and free from difficulties found in previous approaches.
\begin{definition}[Generalized Action Principle]
\lb{AP}
The equation of motion for a physical field $\phi^i$ is the one for which the Action \eqref{H7B} is stationary.
\end{definition}

As a consequence of Definition \ref{AP}, the physical field should satisfy the generalized Euler-Lagrange equation \eqref{GHEL2}. Since for Lagrangian functions independent on the action-density our variational problem \eqref{H7B} reduces to the classical one, our generalized Action Principle is appropriate to describe both conservative and non-conservative systems. In the next section we show, through some examples, that, just like in the original Herglotz problem, our generalized Action for fields describes non-conservative systems when the Lagrangian is linear on the action-density field.

%%%%%%%%%%%%%%%%%%%%%%%%%%%%%%%%%%%%%%%%%%%%%%%%%%%%%%%%%%%%%%%%%%%%%%%%%%%%%%%%%%%%%%%%%%

\section{Examples}\label{sec3}

In this section we show that the generalized Action Principle stated in Definition \ref{AP} enables us to construct simple and physically meaningful Action-dependent Lagrangian functions, which describe a wide range of non-conservative classical and quantum systems.

\subsection{Vibrating String under viscous forces}

In order to illustrate the potential of application of our Action Principle to investigate dissipative systems, we start this section with the simplest example for a non-conservative continuum system in classical mechanics. Let us consider a two-dimensional space-time ($d=2$), with $x_1=t$ and $x_2=x$, the Lagrangian function for a vibrating string under viscous forces can be given by
\beq
\lb{E1}
\mathcal{L}=\frac{\mu}{2}\left(\partial_t \phi\right)^2-T\left(\partial_x \phi\right)^2-\frac{\gamma}{\mu}s^1
\eeq
where $\mu$ is the mass density, $T$ is the tension, $\phi$ is the string displacement, $\gamma$ is the viscous coefficient, and we choose $\gamma_\mu=(\gamma,0)$. As in the problem of a particle under viscous forces, discussed in the beginning of Section \ref{sec2}, the last term in \eqref{E1} can be interpreted as a potential energy for the dissipative force. The first and second terms in \eqref{E1} are the kinetic energy and the elastic potential, respectively.

From the Lagrangian function \eqref{E1}, it is easy to see that our Action Principle gives the correct equation of motion for a string under the presence of a viscous force proportional to the velocity $\partial_t \phi$. By inserting \eqref{E1} into the generalized Euler-Lagrange equation \eqref{GHEL} we get
\beq
\lb{E2}
\mu \partial_{tt} \phi-T\partial_{xx}\phi+\gamma \partial_t \phi=0.
\eeq

\subsection{Non-conservative Electromagnetic Theory}
In this example, we show that a Lagrangian linear in the action-density $s^u$ enables us to formulate a non-conservative electromagnetic theory. Let the position four-vector be $x_\mu=(ct,x,y,z)$ in a Minkowski space with metric $\eta_{\mu\nu}=\mbox{diag}(1,-1,-1,-1)$. We define the Lagrangian density for the electromagnetic field as:
\begin{eqnarray}
\lb{LEMT}
\mathcal{L}=-\dfrac{1}{4}F_{\mu\nu}F^{\mu\nu}-\frac{4\pi}{c}A_\mu J^\mu-\gamma_\mu s^\mu,
\end{eqnarray}
where $F_{\mu\nu}=\partial_{\mu}A_{\nu}-\partial_{\nu}A_{\mu}$ is the standard electromagnetic tensor, $A_\mu=(\phi,A_1,A_2,A_3)$ is the potential four-vector, $J_\mu=(c\rho,J_x,J_y,J_z)=(c\rho,\boldsymbol{J})$ is the four-vector current density, and $\gamma_\mu=(\gamma_0,\gamma_1,\gamma_2,\gamma_3)=(\gamma_0,\boldsymbol{\gamma})$ is a constant four-vector. Thus, from the generalized Euler-Lagrange equation \eqref{GHEL2} we have
\begin{eqnarray}
\partial_\mu F^{\mu\nu}+\gamma_\mu F^{\mu\nu}=\frac{4\pi}{c}J^\nu,
\end{eqnarray}
which corresponds to a generalization for the first pair of Maxwell's equations in the covariant form.
The other pair is automatically obtained from the definition of the tensor $F_{\mu\nu}$ and reads
\begin{equation}
\partial_{[\mu} F_{\nu\lambda]}=\partial_\mu F_{\nu\lambda}+\partial_\nu F_{\lambda\mu}+\partial_\lambda F_{\mu\nu}=0.
\end{equation}
Consequently, the generalized Maxwell's equations for this non-conservative theory is
\beq
\lb{G4ME}
\begin{split}
\nabla\times \mathbf{B}-\frac{1}{c}\partial_t\mathbf{E}+\boldsymbol{\gamma}\times \mathbf{B}-\gamma_0\mathbf{E}&=\frac{4\pi}{c}\mathbf{J}\\
\nabla\cdot\mathbf{E}+\boldsymbol{\gamma}\cdot\mathbf{E}&=4\pi\rho\\
\nabla\times\mathbf{E}-\frac{1}{c}\partial_t\mathbf{B}&=0\\
\nabla\cdot\mathbf{B}&=0.
\end{split}
\eeq
The physical consequences of the constant four-vector $\gamma_\mu$ are evident from the field equations \eqref{G4ME}. In this non-conservative theory the vacuum behaves like a non-standard material medium, where the field induces a virtual charge density $-\frac{1}{4\pi}\boldsymbol{\gamma}\cdot\mathbf{E}$, and virtual current densities $-\frac{c}{4\pi}\boldsymbol{\gamma}\times \mathbf{B}$ and $\frac{c}{4\pi}\gamma_0\mathbf{E}$, respectively. Note that these are virtual charge and current densities since they do not represent real charge displacements in a material medium.

In order to shed light on the non-conservative effects in the electromagnetic theory given by \eqref{LEMT} when $\gamma_{\mu} \neq 0$, it is interesting to investigate the behavior of electromagnetic waves. By choosing the modified gauge
\begin{equation}
\partial_\mu A^\mu+\gamma_\mu A^\mu=0,
\end{equation}
we have the wave equation for the electromagnetic field
\begin{equation}
\square^2F^{\mu\nu}+\gamma_\alpha\partial^\alpha F^{\mu\nu}=\frac{4\pi}{c}(\partial^\mu J^\nu-\partial^\nu J^\mu)=\frac{4\pi}{c}\partial^{[\mu} J^{\nu]},
\label{WE}
\end{equation}
where $\square^ 2=\frac{1}{c^ 2}\partial_t^2-\nabla^ 2$.
This equation clearly describes a damped/forced propagating wave, whose damping/forcing depends on the four-vector $\gamma_\mu$. For simplicity, let us consider only the case where there are no sources ($ J^\mu=0$). In this case, a particular solution for the wave equation  \eqref{WE} is
\begin{equation}
\lb{WE2}
F^{\mu\nu}=F^{\mu\nu}_0e^{ik_\sigma x^\sigma},
\end{equation}
where $F^{\mu\nu}_0$ is a constant antisymmetric tensor with the null principal diagonal, and the four-vector $k_\sigma$ satisfies the relation
\begin{equation}
k_\sigma k^\sigma-i\gamma_\sigma k^\sigma=0\label{e31}.
\end{equation}
In order to better illustrate the non-conservative feature of the field, here we will consider the particular case of a wave propagating in the $z$-direction. Thus,
$F^{\mu\nu}$ is a function of the coordinate $z$ and the time $t$. Another simplification
is to consider that only the time component of the four-vector $\gamma_\mu$ is nonzero, that is, $\boldsymbol{\gamma}=0$ and $\gamma_0\neq 0$. In this case, the wave equation (\ref{WE}) provides three possible solutions for $F^{\mu\nu}$:
\begin{equation}
F_{\mu \nu}(t,z)=
\left\{
\begin{array}{ll}
F_{\mu \nu}^{(\pm)}e^{\frac{-\gamma_0 \pm \gamma^\prime}{2}ct}e^{ikz} & \mbox{if $\gamma_0^2>4k^2$};\\
\left( F_{\mu \nu}^{(+)}+F_{\mu \nu}^{(-)}ct \right) e^{-\frac{\gamma_0}{2}ct}e^{ikz} & \mbox{if $\gamma_0^2=4k^2$};\\
F_{\mu \nu}^{(\pm)}e^{\frac{-\gamma_0 \pm i\gamma^\prime}{2}ct}e^{ikz} & \mbox{if $\gamma_0^2<4k^2$},\\
\end{array}
\right. \nonumber
\end{equation}
where $\gamma^\prime\equiv\sqrt{|\gamma_0^2-4k^2|}$, and $F_{\mu \nu}^{(\pm)}$ are constant antisymmetric tensors, with $F^{(\pm)}_{03}=F^{(\pm)}_{12}=0$.
When $\gamma_0>0$ ($\gamma_0<0$) we observe three cases of damped (forced) waves and, in any of these cases, the amplitude of electromagnetic waves decreases (increases) with time. The first two cases ($\gamma_0^2>4k^2$ and $\gamma_0^2=4k^2$) correspond to stationary waves and occur for small spatial frequencies ($k \le |\gamma_0|/2$), and the third one ($\gamma_0^2<4k^2$) corresponds to traveling waves with velocity $v=\frac{\gamma^\prime}{2 k}c$, smaller than the speed of light $c$.

\subsection{Non-conservative Schr\"{o}dinger equation}
The generalized Action Principle and the resulting Euler-Lagrange equations
can also be applied to formulate a simple and unified Lagrangian function for a dissipative quantum mechanical system with both Stokes (linear on velocity) and Newton (quadratic on velocity) resistances. Let us consider a particle with mass $m$
in a non-conservative system, under the influence of a time-independent potential $V(\mathbf{r})$.
The Lagrange density for this system can be written as
\begin{eqnarray}\label{eq:lag-dens-herglotz}
\mathcal{L}=-\dfrac{\hbar^2}{2m}\nabla\Psi^*(\mathbf{r},t)\nabla\Psi(\mathbf{r},t)
-V(\mathbf{r})\Psi^*(\mathbf{r},t)\Psi(\mathbf{r},t)\\+\dfrac{i\hbar}{2}\left(\Psi^*(\mathbf{r},t)\partial_t\Psi(\mathbf{r},t)-\Psi(\mathbf{r},t)\partial_t\Psi^*(\mathbf{r},t)\right)-\gamma_\mu s^\mu
\end{eqnarray}
where $x_\mu=(t,x,y,z)=(t,\mathbf{r})$, $\partial_\mu=(\partial_t,\partial_x,\partial_y,\partial_z)=(\partial_t,\nabla)$, $\Psi(\mathbf{r},t)$ is the wave function associated to the particle,
$\gamma_\mu=\left(\gamma_0, \gamma_1, \gamma_2, \gamma_3\right)=(\gamma_0,\boldsymbol{\gamma})$ is a constant
four-vector, and $s_\mu=\left(s_0, s_1, s_2, s_3\right)=(s_0,\boldsymbol{s})$
is the action-density four-vector. Applying (\ref{GHEL2}) to (\ref{eq:lag-dens-herglotz}) we obtain the following wave equation,
\begin{equation}\label{eq:Sch-noncons}
i\hbar\partial_t\Psi(\mathbf{r},t)=-\dfrac{\hbar^2}{2m}\nabla^2\Psi(\mathbf{r},t)+V(\mathbf{r})
\Psi(\mathbf{r},t)-i\hbar\dfrac{\gamma_0}{2}\Psi(\mathbf{r},t)-\dfrac{\hbar^2}{2m}\boldsymbol{\gamma}\cdot\nabla
\Psi(\mathbf{r},t)
\end{equation}
which is the Schr\"{o}dinger equation for a non-conservative mechanical system. It is clear that, when $\gamma^\mu$ is a null four-vector, we recover the well known conservative Schr\"odinger wave mechanics. In the particular case where $\gamma_0\neq 0$ and $\boldsymbol{\gamma}=\boldsymbol{0}$, we obtain the wave equation proposed in \cite{Yao:2009}, suitable for a non-conservative force proportional to the velocity, $\mathbf{F}=-k\mathbf{v}$. On the other hand, the case $\gamma_0=0$ and $\boldsymbol{\gamma}\neq\boldsymbol{0}$ corresponds to a force proportional to the square of the velocity \cite{Razavy:2005}.

In order to compare both (conservative and non-conservative) approaches, first we can, in the non-conservative case, look at the resulting continuity equation, which can be easily obtained by considering \eqref{eq:Sch-noncons} and its complex conjugate. This equation then reads
\begin{equation}
\partial_t\rho=-\boldsymbol{\nabla}\cdot \mathbf{J}-\boldsymbol{\gamma}\cdot \mathbf{J}-\gamma_0\rho
\end{equation}
where $\rho(\mathbf{r},t)=\Psi^*(\mathbf{r},t)\Psi(\mathbf{r},t)$ is the probability density and
\begin{equation*}
  \mathbf{J}=\dfrac{\hbar}{2mi}\left(\Psi^*\nabla\Psi-\Psi\nabla\Psi^*\right)
\end{equation*}
is the probability current vector. We can integrate over all space and take into account that the wave function must go to zero in the infinity. Then, in the particular case where $\boldsymbol{\gamma=0}$ we get
\begin{equation}
\int d^3\mathbf{r}\,\rho(\mathbf{r},t)=e^{-\gamma_0t}\int_V d^3\mathbf{r}\,\rho(\mathbf{r},0).
\end{equation}
Therefore, when $\gamma_0>0$ we obtain that the probability of finding the particle in a given region of space (of volume $V$) decays exponentially with time. In other words, dissipation means that the particle escapes from an open system.

\subsection{Non-conservative Klein-Gordon equation (Telegraph Problem)}

We finish our examples with the Klein-Gordon equation, which belongs to a relativistic scalar field theory. A non-conservative formulation for the Klein-Gordon equation can be obtained through the Lagrangian function
\beq
\lb{Ea1}
L= \frac{1}{2}\partial_\mu \phi \partial^\mu \phi -\frac{m^2}{2}\phi^2-\gamma_{\mu} s^{\mu}
,
\eeq
where, as in the electromagnetic problem, $x_\mu=(ct,x,y,z)$ in a Minkowski space with metric $\eta_{\mu\nu}=\mbox{diag}(1,-1,-1,-1)$, $\gamma_\mu=\left(\gamma_0, \gamma_1, \gamma_2, \gamma_3\right)$ is a constant four-vector, and now $\phi$ is a scalar field. From the generalized Euler-Lagrange equation \eqref{GHEL2} we get the equation of motion for the field
\beq
\lb{Ea2}
\left(\partial_\mu\partial^\mu+\gamma_\mu\partial^\mu+m^2\right)\phi=0,
\eeq
which represents a damped/forced Klein-Gordon wave equation. As in the previous examples, it is clear that, when $\gamma_\mu=0$, we recover the classical problem of a conservative relativistic scalar field. Furthermore, in the particular case where $\gamma_\mu=(\gamma_0,0,0,0)$, equation \eqref{Ea2} reduces to the well know Telegraph equation \cite{Courant}
\beq
\lb{Ea3}
\square^2 \phi + \frac{\gamma_0}{c} \partial_{t} \phi+m^2 k\phi =0,
\eeq
where $\square^ 2=\frac{1}{c^ 2}\partial_t^2-\nabla^ 2$.

%%%%%%%%%%%%%%%%%%%%%%%%%%%%%%%%%%%%%%%%%%%%%%%%%%%%%%%%%%%%%%%%%%%%%%%%%%%%%%%%%%%%%%%%%%

\section{Conclusions}\label{sec4}

In this work, we generalized the ideas introduced in \cite{MJJG} and formulated an Action Principle with Action-dependent Lagrangian that describes non-conservative systems.
In contrast to the case of gravitation \cite{MJJG}, in the present proposal, we obtained a generalized Euler-Lagrange equation for arbitrary Action-dependent Lagrangian functions. Differently, from other approaches found in the literature, our generalized Action Principle enables us to construct meaningful Lagrangian functions, which provide physically consistent expressions for the momentum and the Hamiltonian of the system.
We have shown, by working out some examples, that the simplest case of a Lagrangian linear on the action-density provides the correct equation of motion for dissipative systems, described by both classical and quantum mechanics. Our present results, as well as those obtained in the case of gravitation \cite{MJJG}, illustrates the potential of application of our generalized Action Principle in the study of a variety of non-conservative systems. Finally, in which concerns to developments and applications of this generalized Action Principle, there are many directions of investigation left to explore. As some interesting examples, we have the quantization of Action-dependent Lagrangian functions and the generalization of the Noether's theorem, in order to investigate conservative laws for non-conservative systems. These and other examples are left to future works.

%%%%%%%%%%%%%%%%%%%%%%%%%%%%%%%%%%%%%%%%%%%%%%%%%%%%%%%%%%%%%%%%%%%%%%%%%%%%%%%%%%%%%%%%%%

\section*{Acknowledgments}
This work was partially supported by CNPq and CAPES (Brazilian research funding agencies).

%%%%%%%%%%%%%%%%%%%%%%%%%%%%%%%%%%%%%%%%%%%%%%%%%%%%%%%%%%%%%%%%%%%%%%%%%%%%%%%%%%%%%%%%%%


\begin{thebibliography}{}

%Introduction

\bibitem{bauer} P. S. Bauer, Proc. Natl. Acad. Sci {\bf 17} 311 (1931).

\bibitem{Stevens}  K. W. H. Stevens, Proc. Phys. Soc. London {\bf 72} 1027 (1958); P. Havas, Nuovo Cimento Suppl. {\bf 5} 363 (1957); F. Negro and A. Tartaglia, Phys. Lett. {\bf 77A} 1 (1980); Phys. Rev. A {\bf 23} 1591 (1981); J. R. Brinati and S. S. Mizrahi, J. Math. Phys. {\bf 21} 2154 (1980); A. Tartaglia, Eur. J. Phys. {\bf 4} 231 (1983).

\bibitem{Morse}  H. Bateman, Phys. Rev. {\bf 38} 815, (1931); P. M. Morse and H. Feshbach, {\it Methods of Theoretical Physics}, McGraw-Hill, New York, (1953), pp. 298-299; H. Feshbach and Y. Tikochinsky, Trans. N. Y. Acad. Sci. {\bf 38} 44 (1977); E. Celeghini, M. Rasetti, M. Tarlini and G. Vitiello, Mod. Phys. Lett. B {\bf 3} 1213 (1989); E. Celeghini, H. Rasetti and G. Vitiello, Ann. Phys. (N.Y.) {\bf 215} 156 (1992).

\bibitem{VujaJones} B. D. Vujanovic and S. E. Jones, {\it Variational Methods in Nonconservative Phenomena}, Academic, San Diego (1989).

\bibitem{Riewe} F. Riewe, Phys. Rev. E {\bf 53} 1890 (1996).

\bibitem{LazoCesar} M. J. Lazo and C. E. Krumreich,
%The action principle for dissipative systems,
J. Math. Phys. {\bf 55}, 122902 (2014).

\bibitem{Herg1}  G. Herglotz, Ber\"uhrungstransformationen, Lectures at the University of G\"ottingen,
G\"ottingen, (1930).

\bibitem{Herg2}  R. B. Guenther, C. M. Guenther and J. A. Gottsch, The Herglotz Lectures
on Contact Transformations and Hamiltonian Systems, Lecture Notes in Nonlinear Analysis,
Vol. 1, Juliusz Schauder Center for Nonlinear Studies, Nicholas Copernicus University, Tor\'un, (1996).

\bibitem{GGB} B. Georgieva, R. Guenther and T. Bodurov,
%Generalized variational principle of Herglotz for several independent variables. First Noether-type theorem,
J. Math. Phys. {\bf 44}, 3911 (2003).

\bibitem{MJJG} M. J. Lazo, J. Paiva, J. T. S. Amaral, G. S. F. Frederico,
%Action principle for action-dependent Lagrangians toward nonconservative gravity: Accelerating universe without dark energy,
Phys. Rev. D {\bf 95}, 101501(R) (2017).

\bibitem{Yao:2009} Xiang-Yao Wu, Bai-Jun Zhang, Hai-Bo Li, Xiao-Jing Liu, Jing-Wu Li and Yi-Qing Guo, Int. J. Theor. Phys. {\bf 48}, 2027 (2009).

\bibitem{Razavy:2005} M. Razavy, {\it Classical and Quantum Dissipative Systems}, Imperial College Press, London
(2005).

\bibitem{Courant} R. Courant and D. Hilbert, {\it Methods of Mathematical Physics, Vol. 2}, Wiley-VCH; Singapore (1989).




\end{thebibliography}
\end{document}